\documentclass[runningheads]{llncs}

\usepackage{amsmath,amssymb,mathtools}
\usepackage{graphicx}

\usepackage[ruled,vlined]{algorithm2e}
\DontPrintSemicolon
\SetKwInOut{KwIn}{Input}
\SetKwInOut{KwOut}{Output}
\SetAlgoInsideSkip{smallskip}
\SetKwComment{tcc}{// }{}

\usepackage{tikz}
\usetikzlibrary{positioning,shapes,arrows,calc}

\usepackage{enumitem}
\setlist[itemize]{leftmargin=*, itemsep=2pt, topsep=2pt, parsep=0pt, partopsep=0pt}
\setlist[enumerate]{leftmargin=*, itemsep=2pt, topsep=2pt, parsep=0pt, partopsep=0pt}

\usepackage[hidelinks]{hyperref}

\begin{document}
%=========================

\title{Permutation Matching Under Parikh Budgets: Linear-Time Detection, Packing, and Disjoint Selection}
\titlerunning{Permutation Matching Under Parikh Budgets}

% --- Authors (edit as needed) ---
\author{MD Nazmul Alam Shanto \and Md. Tanzeem Rahat \and Md. Manzurul Hasan}
\authorrunning{M. N. A. Shanto et al.}
\institute{Department of Computer Science, American International University-Bangladesh, Dhaka, Bangladesh\\
\email{mdnazmulshanto2001@gmail.com, tanzeem.rahat@aiub.edu, manzurul@aiub.edu}}

\maketitle

%=========================
\begin{abstract}
We study permutation (jumbled/Abelian) pattern matching over a general alphabet $\Sigma$. Given a pattern $P$ of length $m$ and a text $T$ of length $n$, the classical task is to decide whether $T$ contains a length-$m$ substring whose Parikh vector equals that of $P$. While this existence problem admits a linear-time sliding-window solution, many practical applications require optimization and packing variants beyond mere detection. We present a unified sliding-window framework based on maintaining the Parikh-vector difference between $P$ and the current window of $T$, enabling permutation matching in $O(n+\sigma)$ time and $O(\sigma)$ space, where $\sigma = |\Sigma|$. Building on this foundation, we introduce a combinatorial-optimization variant that we call Maximum Feasible Substring under Pattern Supply (MFSP): find the longest substring $S$ of $T$ whose symbol counts are component-wise bounded by those of $P$. We show that MFSP can also be solved in $O(n+\sigma)$ time via a two-pointer feasibility maintenance algorithm, providing an exact packing interpretation of $P$ as a resource budget.
Finally, we address non-overlapping occurrence selection by modeling each permutation match as an equal-length interval and proving that a greedy earliest-finishing strategy yields a
maximum-cardinality set of disjoint matches, computable in linear time once all matches are enumerated. Our results provide concise, provably correct algorithms with tight bounds, and connect
frequency-based string matching to packing-style optimization primitives.

\end{abstract}

%=========================
\section{Introduction}
\label{sec:intro}

Pattern matching is a foundational task in string algorithms with applications spanning
text processing, bioinformatics, cybersecurity, and information retrieval.  Beyond the
classical setting-where a pattern must appear as a substring with the \emph{same order}-a
widely studied variant asks whether a pattern occurs \emph{up to permutation}.  Concretely,
given a pattern $P$ of length $m$ and a text $T$ of length $n$, the \emph{permutation (jumbled/Abelian) matching}
problem asks whether $T$ contains a substring of length $m$ whose multiset of symbols equals
that of $P$.  This problem is also known as \emph{jumbled pattern matching}~\cite{cicalese2009} and is closely related
to permutation pattern matching for permutations~\cite{ibarra1997,bose1998}.

Permutation matching has been studied under multiple lenses.
In the binary case, specialized indexing and algorithms exploit one-dimensional summaries
of windows~\cite{moosa2010,Giaquinta2013}.  For general alphabets, the central object is the
\emph{Parikh vector} (frequency vector) of a string, and jumbled matching becomes a question of
frequency equality~\cite{cicalese2009,Burcsi2011}.  A large body of work addresses \emph{indexing}
for answering many queries efficiently, including efficient indexes for constant-sized alphabets~\cite{Kociumaka2016},
hardness results for general jumbled indexing~\cite{Amir2014}, approximate variants~\cite{Burcsi2012},
online algorithms~\cite{Ghuman2018}, SIMD acceleration~\cite{Ghuman2016}, and even quantum models~\cite{Juarez2022}.
These results underline both the theoretical richness and the practical relevance of frequency-based matching.

\medskip
\noindent \emph{From detection to optimization and packing.}
While detecting a permutation occurrence is fundamental, many applications demand richer
\emph{optimization} and \emph{packing} primitives.  For example: (i) when the pattern acts as a \emph{resource budget}
(e.g., allowed counts of tokens), one may seek the \emph{longest} substring of $T$ that does not exceed this budget;
(ii) when multiple occurrences exist, one may want a maximum-size set of \emph{non-overlapping} matches.  Such questions
naturally connect jumbled matching to combinatorial optimization: feasibility under a supply vector resembles a packing constraint,
and selecting disjoint occurrences becomes an interval packing problem.

\begin{figure}[t]
\centering
\begin{tikzpicture}[scale=0.95, every node/.style={font=\small}]
% --- Top: sliding window anagram check ---
\node[draw, rounded corners, inner sep=3pt] (T) at (0,0)
{$T=\; \boxed{a}\;\boxed{b}\;\boxed{c}\;\boxed{a}\;\boxed{b}\;\boxed{d}\;\boxed{c}\;\boxed{b}$};

\draw[thick, red, rounded corners]
(-1.7,-0.38) rectangle (-0.1,0.38);
\node[red] at (-0.6,0.65) {window $W_i$};

\node[draw, rounded corners, inner sep=3pt] (P) at (0,-1.15)
{$P=\; \boxed{b}\;\boxed{a}\;\boxed{c}$};

\draw[->, thick] (-0.6,-0.38) -- (-0.6,-0.95);

\node[align=left] at (3.1,-1.1)
{$\mathrm{freq}(W_i)=\mathrm{freq}(P)$?\\
(Parikh equality)};

% --- Bottom: supply feasibility view ---
\node[draw, rounded corners, inner sep=3pt] (S) at (0,-2.55)
{$S=\;T[\ell..r]$};

\node[align=left] at (5.2,-2.55)
{$\mathrm{freq}(S)\preceq \mathrm{freq}(P)$\\
(component-wise budget / packing feasibility)};

\end{tikzpicture}
\caption{Two viewpoints used in this paper. Top: permutation matching asks whether some length-$m$ window $W_i$
has the same Parikh vector as $P$. Bottom: treating $\mathrm{freq}(P)$ as a supply vector yields a packing-style feasibility constraint
for selecting long substrings under symbol budgets.}
\label{fig:intro}
\end{figure}
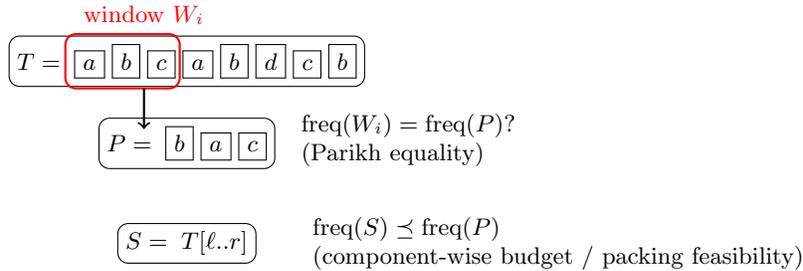

\medskip
\noindent\emph{Our contributions.}
This paper presents a concise sliding-window framework with provable guarantees, and then
builds optimization/packing variants that are natural in combinatorial optimization settings.
\begin{itemize}
  \item \emph{Linear-time permutation matching (baseline).}
  We present an $O(n+\sigma)$-time, $O(\sigma)$-space algorithm for deciding whether $T$ contains a length-$m$
  permutation occurrence of $P$ by maintaining the Parikh-difference vector and a nonzero counter.

  \item \emph{MFSP: Maximum Feasible Substring under Pattern Supply.}
  We introduce an optimization variant: find the longest substring $S$ of $T$ such that
  $\mathrm{freq}(S)\preceq \mathrm{freq}(P)$.  Interpreting $\mathrm{freq}(P)$ as a resource budget connects
  the problem to packing feasibility.  We give a linear-time two-pointer algorithm with a full correctness proof.

  \item \emph{Non-overlapping occurrence selection via interval packing.}
  When all permutation matches are viewed as equal-length intervals, selecting a maximum number of disjoint matches
  becomes an interval packing problem.  We prove that a greedy earliest-finishing (equivalently earliest-start) rule
  is optimal and yields a linear-time procedure once matches are enumerated.
\end{itemize}

\noindent\emph{Organization.}
Section~\ref{sec:prelim} introduces notation and Parikh vectors.
Section~\ref{sec:perm} presents the incremental Parikh-difference algorithm for permutation matching.
Section~\ref{sec:mfsp} formulates MFSP and proves the correctness of a linear-time feasibility-maintenance algorithm.
Section~\ref{sec:packing} studies non-overlapping selection as interval packing.
Section~\ref{sec:exp} reports experimental observations, and
Section~\ref{sec:conc} concludes.

%=========================
\section{Preliminaries and Notation}\label{sec:prelim}

Let $\Sigma$ be an alphabet of size $\sigma = |\Sigma|$.  A string (word) $S$ over $\Sigma$ has
length $|S|$, and we use $0$-based indexing so that $S[i]$ denotes the character at position
$i$ for $0 \le i < |S|$.  For integers $\ell \le r$, we write $S[\ell..r]$ for the substring
$S[\ell]S[\ell+1]\cdots S[r]$.  For a text $T$ of length $n$ and an integer $m \le n$, we
denote by
\[
W_i \;=\; T[i..i+m-1] \qquad (0 \le i \le n-m)
\]
the length-$m$ window of $T$ starting at index $i$.

\subsection{Parikh (frequency) vectors}
For a string $S$, the \emph{Parikh vector} (frequency vector) of $S$ is the function
$\mathrm{freq}_S : \Sigma \to \mathbb{Z}_{\ge 0}$ defined by
\[
\mathrm{freq}_S(c) \;=\; |\{\,i : 0 \le i < |S|,\; S[i]=c \,\}| \qquad \text{for each } c \in \Sigma.
\]
When the context is clear, we write $\mathrm{freq}(S)$ for $\mathrm{freq}_S$.
Two strings $X$ and $Y$ are \emph{permutations} (anagrams) of each other if $|X|=|Y|$ and
$\mathrm{freq}(X)=\mathrm{freq}(Y)$ (component-wise equality).

\subsection{Problems studied}
Throughout, we are given a pattern $P \in \Sigma^m$ and a text $T \in \Sigma^n$.

\paragraph{Permutation matching (existence).}
Decide whether there exists an index $i$ with $0 \le i \le n-m$ such that
\[
\mathrm{freq}(W_i) \;=\; \mathrm{freq}(P).
\]

\paragraph{MFSP (maximum feasible substring under pattern supply).}
Treat $\mathrm{freq}(P)$ as a \emph{supply} (budget) vector.  A substring $S$ of $T$ is
\emph{feasible} if its symbol counts do not exceed the supply:
\[
\mathrm{freq}(S) \;\preceq\; \mathrm{freq}(P)
\quad\Longleftrightarrow\quad
\forall c\in\Sigma:\; \mathrm{freq}(S)(c) \le \mathrm{freq}(P)(c).
\]
MFSP asks to compute a longest feasible substring of $T$, i.e.,
\[
\max\{\,|S| : S = T[\ell..r] \text{ and } \mathrm{freq}(S)\preceq \mathrm{freq}(P)\,\}.
\]

\paragraph{Non-overlapping selection.}
Let $\mathcal{M}=\{\, i : \mathrm{freq}(W_i)=\mathrm{freq}(P)\,\}$ be the set of all match
starting positions.  Each $i\in\mathcal{M}$ corresponds to an interval
$I_i=[i,\, i+m-1]$.  The goal is to select a maximum-cardinality subset of pairwise
disjoint intervals.

\subsection{Difference vectors and a nonzero counter}
A central tool in this paper is maintaining \emph{differences} of Parikh vectors.
For a window $W_i$ and pattern $P$, define the difference vector
\[
\Delta_i(c) \;=\; \mathrm{freq}(W_i)(c) - \mathrm{freq}(P)(c)
\qquad \text{for } c\in\Sigma.
\]
Observe that $\mathrm{freq}(W_i)=\mathrm{freq}(P)$ if and only if $\Delta_i(c)=0$ for all
$c\in\Sigma$.  To test this efficiently we maintain
\[
nz_i \;=\; |\{\,c\in\Sigma : \Delta_i(c) \ne 0 \,\}|,
\]
the number of nonzero entries.  Then $W_i$ is a permutation match if and only if $nz_i=0$.

\subsection{Implementation model}
We assume either (i) $\Sigma$ is given explicitly and $\sigma$ is moderate, enabling
array-based counts indexed by symbols, or (ii) we perform \emph{alphabet compression}:
map only the symbols that appear in $P$ and the relevant portion of $T$ to integers
$1,\dots,\sigma'$ using a hash map, where $\sigma' \le \min(\sigma, m+n)$.
All asymptotic bounds are stated in terms of $\sigma$ for simplicity, and extend directly
to $\sigma'$ under compression.

\begin{figure}[t]
\centering
\begin{tikzpicture}[scale=1.0, every node/.style={font=\small}]
% Text row
\node (t0) at (0,0) {$T:$};
\node[draw, circle, inner sep=1.8pt] (a0) at (0.8,0) {$a$};
\node[draw, circle, inner sep=1.8pt] (b0) at (1.6,0) {$b$};
\node[draw, circle, inner sep=1.8pt] (c0) at (2.4,0) {$c$};
\node[draw, circle, inner sep=1.8pt] (a1) at (3.2,0) {$a$};
\node[draw, circle, inner sep=1.8pt] (b1) at (4.0,0) {$b$};
\node[draw, circle, inner sep=1.8pt] (d1) at (4.8,0) {$d$};

% Window highlight
\draw[thick, red, rounded corners] (0.45,-0.43) rectangle (2.75,0.5);
\node[red] at (1.6,0.65) {$W_0$};

\draw[thick, blue, rounded corners] (1.25,-0.35) rectangle (3.55,0.35);
\node[blue] at (2,-0.75) {$W_1$};

% Pattern row
\node (p0) at (0,-1.6) {$P:$};
\node[draw, rectangle, inner sep=2pt] (pb) at (0.8,-1.6) {$b$};
\node[draw, rectangle, inner sep=2pt] (pa) at (1.6,-1.6) {$a$};
\node[draw, rectangle, inner sep=2pt] (pc) at (2.4,-1.6) {$c$};

% Arrow
\draw[->, thick] (2.4,-0.15) -- (2.4,-1.25);

\node[align=left] at (5.9,-1.5)
{$\Delta_i = \mathrm{freq}(W_i)-\mathrm{freq}(P)$\\
Match $\Longleftrightarrow nz_i=0$};
\end{tikzpicture}
\caption{Notation: a length-$m$ window $W_i$ slides over $T$.  Maintaining the Parikh-difference vector
$\Delta_i$ and the nonzero counter $nz_i$ allows constant-time updates per shift.}
\label{fig:prelim-diff}
\end{figure}
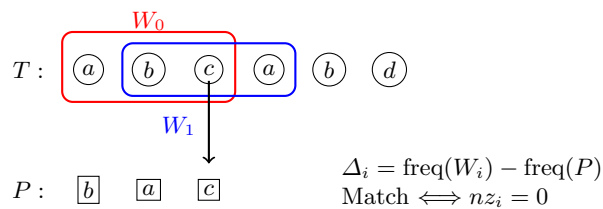

%=========================
\section{Linear-Time Permutation Matching via Parikh-Difference Maintenance}\label{sec:perm}

We first address the classical \emph{permutation (jumbled/Abelian) matching} decision problem:
given a pattern $P \in \Sigma^m$ and a text $T \in \Sigma^n$ with $m \le n$, decide whether
there exists an index $i$ such that the length-$m$ window $W_i = T[i..i+m-1]$ satisfies
$\mathrm{freq}(W_i)=\mathrm{freq}(P)$.  The key is to maintain, for the current window, the
Parikh-difference vector and a single integer that certifies equality.

\subsection{Maintained state}
Recall the difference vector
\[
\Delta_i(c) \;=\; \mathrm{freq}(W_i)(c)-\mathrm{freq}(P)(c), \qquad c\in\Sigma,
\]
and the \emph{nonzero counter}
\[
nz_i \;=\; |\{\,c\in\Sigma : \Delta_i(c)\neq 0\,\}|.
\]
Then $W_i$ is a permutation match if and only if $nz_i=0$.

To update $nz_i$ in constant time, we use the following primitive: when changing
$\Delta(c)$ by $\pm 1$, only the status of $\Delta(c)$ being zero/nonzero can change.

\medskip
\noindent\emph{Update primitive.}
Let $x=\Delta(c)$ before an update and $x'=\Delta(c)+\delta$ after the update, where
$\delta\in\{-1,+1\}$.
\begin{itemize}
    \item If $x=0$ and $x'\neq 0$, then increment $nz$.
    \item If $x\neq 0$ and $x'=0$, then decrement $nz$.
    \item Otherwise, $nz$ is unchanged.
\end{itemize}

\subsection{Algorithm}
We initialize $\Delta$ by starting from $-\mathrm{freq}(P)$ and adding the symbols of the first
window $W_0$.  After that, each shift from $W_i$ to $W_{i+1}$ removes one outgoing symbol
$T[i]$ and adds one incoming symbol $T[i+m]$; hence only two coordinates of $\Delta$ change
per step.

\begin{algorithm}[t]
\caption{Permutation Matching by Parikh-Difference Maintenance}\label{alg:perm}
\KwIn{Text $T$ of length $n$, pattern $P$ of length $m$ $(m\le n)$, alphabet $\Sigma$}
\KwOut{\textsf{True} iff some window $W_i$ is a permutation of $P$ (optionally all match indices)}

\BlankLine
\textbf{Initialize:} For all $c\in\Sigma$, set $\Delta(c)\gets 0$ and $nz\gets 0$.\\
\ForEach{$c\in\Sigma$}{
    \tcp{Start from $-\mathrm{freq}(P)$}
    \If{$\mathrm{freq}(P)(c)>0$}{
        $\Delta(c)\gets -\mathrm{freq}(P)(c)$;\;
        $nz\gets nz+1$ \tcp*{since $\Delta(c)\neq 0$}
    }
}
\BlankLine
\For{$j\gets 0$ \KwTo $m-1$}{
    \tcp{Add the first window $W_0$}
    $x\gets \Delta(T[j])$;\;
    $\Delta(T[j])\gets \Delta(T[j])+1$;\;
    \If{$x=0$ \textbf{and} $\Delta(T[j])\neq 0$}{ $nz\gets nz+1$ }
    \ElseIf{$x\neq 0$ \textbf{and} $\Delta(T[j])=0$}{ $nz\gets nz-1$ }
}
\If{$nz=0$}{\Return \textsf{True}}

\BlankLine
\For{$i\gets 0$ \KwTo $n-m-1$}{
    \tcp{Shift from $W_i$ to $W_{i+1}$: remove outgoing, add incoming}
    $out\gets T[i]$;\;\; $in\gets T[i+m]$;\;

    $x\gets \Delta(out)$;\;
    $\Delta(out)\gets \Delta(out)-1$;\;
    \If{$x=0$ \textbf{and} $\Delta(out)\neq 0$}{ $nz\gets nz+1$ }
    \ElseIf{$x\neq 0$ \textbf{and} $\Delta(out)=0$}{ $nz\gets nz-1$ }

    $x\gets \Delta(in)$;\;
    $\Delta(in)\gets \Delta(in)+1$;\;
    \If{$x=0$ \textbf{and} $\Delta(in)\neq 0$}{ $nz\gets nz+1$ }
    \ElseIf{$x\neq 0$ \textbf{and} $\Delta(in)=0$}{ $nz\gets nz-1$ }

    \If{$nz=0$}{\Return \textsf{True}}
}
\Return \textsf{False}
\end{algorithm}

\subsection{Correctness}
\begin{theorem}\label{thm:perm-correct}
Algorithm~\ref{alg:perm} returns \textsf{True} if and only if there exists an index
$i$ with $0\le i\le n-m$ such that $\mathrm{freq}(W_i)=\mathrm{freq}(P)$.
\end{theorem}

\begin{proof}
We prove the invariant that after initialization and after each shift, for every $c\in\Sigma$,
the maintained value $\Delta(c)$ equals $\mathrm{freq}(W_i)(c)-\mathrm{freq}(P)(c)$ for the
current window $W_i$, and $nz$ equals the number of symbols with nonzero $\Delta(c)$.

Initialization sets $\Delta(c)=-\mathrm{freq}(P)(c)$ (or $0$ if absent), then adds each symbol
of $W_0$ once; thus $\Delta(c)=\mathrm{freq}(W_0)(c)-\mathrm{freq}(P)(c)$.  The update rules
modify $nz$ exactly when a coordinate crosses zero, hence $nz=|\{c:\Delta(c)\neq 0\}|$ holds.

For a shift from $W_i$ to $W_{i+1}$, exactly one symbol $out=T[i]$ leaves the window and one
symbol $in=T[i+m]$ enters. Therefore, for all $c\in\Sigma$,
$\mathrm{freq}(W_{i+1})(c)=\mathrm{freq}(W_i)(c)-\mathbf{1}[c=out]+\mathbf{1}[c=in]$.
The algorithm applies precisely these two $\pm 1$ updates to $\Delta$, preserving the
difference-vector invariant, and updates $nz$ consistently with zero-crossings.

Finally, $\mathrm{freq}(W_i)=\mathrm{freq}(P)$ holds if and only if
$\Delta(c)=0$ for all $c$, which holds if and only if $nz=0$.  Hence the algorithm returns
\textsf{True} exactly for those windows that are permutations of $P$.
\end{proof}

\subsection{Complexity}
\begin{theorem}\label{thm:perm-time}
Algorithm~\ref{alg:perm} runs in $O(n+\sigma)$ time and uses $O(\sigma)$ space.
\end{theorem}

\begin{proof}
The initialization over $\Sigma$ takes $O(\sigma)$ time.  Building the first window performs
$m$ constant-time updates.  Each of the $n-m$ shifts performs two constant-time updates and
one constant-time check of $nz$, for total $O(n)$ time.  The maintained arrays/maps store one
integer per symbol in $\Sigma$, hence $O(\sigma)$ space.
\end{proof}

\medskip
\noindent\emph{Remark (enumerating all matches).}
With the same maintained state, one can output all match indices by recording every $i$
for which $nz=0$; this adds $O(\#\mathcal{M})$ output time and does not change the asymptotic
running time.

%=========================
\section{MFSP: Maximum Feasible Substring under Pattern Supply}\label{sec:mfsp}

We now move from detection to an optimization primitive that naturally fits a packing
interpretation.  Let $P\in\Sigma^m$ be a \emph{supply string}.  Its Parikh vector
$\mathrm{freq}(P)$ defines a component-wise budget on the symbols that may appear in a
chosen substring of the text $T$.

\subsection{Problem definition}
A substring $S=T[\ell..r]$ is \emph{feasible} if
\[
\mathrm{freq}(S) \preceq \mathrm{freq}(P)
\quad\Longleftrightarrow\quad
\forall c\in\Sigma:\; \mathrm{freq}(S)(c) \le \mathrm{freq}(P)(c).
\]
\noindent\emph{MFSP} (Maximum Feasible Substring under Pattern Supply) asks to find
\[
(\ell^\star,r^\star)\in \arg\max_{0\le \ell\le r < n}\; (r-\ell+1)
\quad \text{s.t.}\quad \mathrm{freq}(T[\ell..r]) \preceq \mathrm{freq}(P).
\]
We also output the maximum length $L^\star=r^\star-\ell^\star+1$.

\medskip
\noindent\emph{Packing viewpoint.}
Interpreting $\mathrm{freq}(P)$ as a budget vector, a substring $T[\ell..r]$ is feasible if
its symbol-count consumption does not exceed the available supply.  MFSP therefore asks for
the \emph{largest feasible packed block} (a contiguous selection) within $T$.

\subsection{Two-pointer feasibility maintenance}
Feasibility is monotone under shrinking: if a window is infeasible, removing characters from
its left end can only decrease its counts and eventually restore feasibility.  This monotonicity
supports a linear-time two-pointer algorithm (a.k.a.\ sliding window) that maintains a maximal
feasible window for each right endpoint.

We maintain an array $\mathrm{cnt}(c)$ for the current window $T[\ell..r]$ and the supply
array $\mathrm{sup}(c)=\mathrm{freq}(P)(c)$.  The window is feasible iff
$\mathrm{cnt}(c)\le \mathrm{sup}(c)$ for all $c$.  To test feasibility in $O(1)$ amortized time,
we maintain a violation counter
\[
viol \;=\; |\{\,c\in\Sigma : \mathrm{cnt}(c) > \mathrm{sup}(c)\,\}|.
\]
As with $nz$ in Section~\ref{sec:perm}, updating $\mathrm{cnt}(c)$ by $\pm 1$ only changes
whether a symbol is violating at a threshold crossing.

\begin{algorithm}[t]
\caption{MFSP by Two-Pointer Supply Feasibility}\label{alg:mfsp}
\KwIn{Text $T$ of length $n$, supply pattern $P$, alphabet $\Sigma$}
\KwOut{A longest feasible substring $T[\ell^\star..r^\star]$ and its length $L^\star$}

\BlankLine
Compute $\mathrm{sup}(c)\gets \mathrm{freq}(P)(c)$ for all $c\in\Sigma$.\\
Initialize $\ell\gets 0$, $L^\star\gets 0$, $(\ell^\star,r^\star)\gets (0,-1)$.\\
For all $c\in\Sigma$: $\mathrm{cnt}(c)\gets 0$;\; $viol\gets 0$.\\

\BlankLine
\For{$r\gets 0$ \KwTo $n-1$}{
    $c\gets T[r]$;\;
    $x\gets \mathrm{cnt}(c)$;\;
    $\mathrm{cnt}(c)\gets \mathrm{cnt}(c)+1$;\;
    \If{$x \le \mathrm{sup}(c)$ \textbf{and} $\mathrm{cnt}(c) > \mathrm{sup}(c)$}{
        $viol\gets viol+1$ \tcp*{crossed into violation}
    }

    \While{$viol>0$}{
        $d\gets T[\ell]$;\;
        $y\gets \mathrm{cnt}(d)$;\;
        $\mathrm{cnt}(d)\gets \mathrm{cnt}(d)-1$;\;
        \If{$y > \mathrm{sup}(d)$ \textbf{and} $\mathrm{cnt}(d) \le \mathrm{sup}(d)$}{
            $viol\gets viol-1$ \tcp*{resolved violation}
        }
        $\ell\gets \ell+1$;\;
    }

    \tcp{Now $T[\ell..r]$ is feasible and is maximal for this $r$ under the maintained $\ell$}
    \If{$r-\ell+1 > L^\star$}{
        $L^\star \gets r-\ell+1$;\;
        $(\ell^\star,r^\star)\gets (\ell,r)$;\;
    }
}
\Return{$(\ell^\star,r^\star), L^\star$}
\end{algorithm}

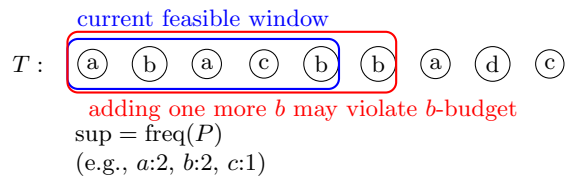
\begin{figure}[t]
\centering
\begin{tikzpicture}[scale=0.95, every node/.style={font=\small}]
% Text row
\node (t) at (0,0) {$T:$};
\foreach \x/\ch in {0.9/a,1.7/b,2.5/a,3.3/c,4.1/b,4.9/b,5.7/a,6.5/d,7.3/c}{
  \node[draw, circle, inner sep=1.8pt] at (\x,0) {\ch};
}
% Supply
\node[align=left] at (2.0,-1.2) {$\mathrm{sup}=\mathrm{freq}(P)$\\
(e.g., $a{:}2,\,b{:}2,\,c{:}1$)};
% Window highlight
\draw[thick, blue, rounded corners] (0.55,-0.35) rectangle (4.35,0.35);
\node[blue] at (2.45,0.65) {current feasible window};

% Violation illustration
\draw[thick, red, rounded corners] (0.55,0.45) rectangle (5.15,-0.4);
\node[red] at (3.85,-0.65) {adding one more $b$ may violate $b$-budget};

\end{tikzpicture}
\caption{MFSP maintains a feasible window under the supply vector $\mathrm{sup}=\mathrm{freq}(P)$.
When the window becomes infeasible, the left pointer advances until feasibility is restored.}
\label{fig:mfsp}
\end{figure}

\subsection{Correctness}
We prove that Algorithm~\ref{alg:mfsp} returns a longest feasible substring.

\begin{theorem}\label{thm:mfsp-correct}
Algorithm~\ref{alg:mfsp} outputs indices $(\ell^\star,r^\star)$ such that
$T[\ell^\star..r^\star]$ is feasible and has maximum length among all feasible substrings of $T$.
\end{theorem}

\begin{proof}
We maintain the invariant that after the \textbf{while}-loop terminates for a fixed right
endpoint $r$, the current window $T[\ell..r]$ is feasible, i.e., for all $c\in\Sigma$,
$\mathrm{cnt}(c)\le \mathrm{sup}(c)$.

\smallskip\noindent
\emph{Feasibility restoration.}
When $T[r]$ is added, the only possible new infeasibility is that for the character
$c=T[r]$, $\mathrm{cnt}(c)$ exceeds $\mathrm{sup}(c)$, which is detected by incrementing
$viol$. Each iteration of the \textbf{while}-loop removes exactly one character from the
left, decreasing exactly one count by $1$. A violation for some symbol $d$ disappears
precisely when its count crosses from $\mathrm{cnt}(d)>\mathrm{sup}(d)$ to
$\mathrm{cnt}(d)\le\mathrm{sup}(d)$, at which point we decrement $viol$. Thus the loop stops
exactly when there are no violating symbols, i.e., feasibility holds.

\smallskip\noindent
\emph{Maximality for each $r$.}
Fix $r$. Let $\ell_r$ be the value of $\ell$ after the while-loop ends. At that moment,
$T[\ell_r..r]$ is feasible. Moreover, the algorithm maintains $\ell$ as small as possible
subject to feasibility: if $\ell$ were decreased by one (i.e., if we tried to include
$T[\ell_r-1]$), then that character was removed at some point only because the window was
infeasible; hence re-including it would reintroduce a violation. Therefore, among all
feasible substrings ending at $r$, the window $T[\ell_r..r]$ has maximum length.

\smallskip\noindent
\emph{Global optimality.}
Let $S^\star=T[a..b]$ be a longest feasible substring in $T$. Consider the iteration
$r=b$. By maximality for this $r$, the algorithm's feasible window ending at $b$ has length
at least $|T[a..b]|$. Therefore, the maximum length $L^\star$ recorded by the algorithm
over all right endpoints is at least $|S^\star|$. Since the algorithm only records lengths
of feasible windows, $L^\star$ cannot exceed the optimum. Hence $L^\star=|S^\star|$ and the
recorded $(\ell^\star,r^\star)$ is optimal.
\end{proof}

\subsection{Complexity}
\begin{theorem}\label{thm:mfsp-time}
Algorithm~\ref{alg:mfsp} runs in $O(n+\sigma)$ time and uses $O(\sigma)$ space.
\end{theorem}

\begin{proof}
Each step increments $r$ once, so the outer loop runs $n$ times. The pointer $\ell$ only
moves forward; across the entire execution it advances at most $n$ times. Each advancement
performs $O(1)$ updates. Thus the total time spent in the while-loops is $O(n)$, and the
total running time is $O(n+\sigma)$ including initialization of arrays over $\Sigma$.
The algorithm stores $\mathrm{sup}$ and $\mathrm{cnt}$ arrays (or compressed maps), requiring
$O(\sigma)$ space.
\end{proof}

\medskip
\noindent\emph{Remark (recovering the substring).}
The algorithm already stores $(\ell^\star,r^\star)$, so the optimal substring is
$T[\ell^\star..r^\star]$ and can be output in $O(L^\star)$ time if desired.

%=========================
\section{Non-overlapping Permutation Occurrences as Interval Packing}\label{sec:packing}

This section addresses selecting a maximum number of \emph{non-overlapping} permutation
occurrences of the pattern $P$ in the text $T$.  Once all permutation matches are identified,
each match corresponds to a fixed-length interval, and the selection task becomes a
classical interval packing problem.

\subsection{From matches to intervals}
Let
\[
\mathcal{M} \;=\; \{\, i \in \{0,\dots,n-m\} : \mathrm{freq}(W_i)=\mathrm{freq}(P)\,\}
\]
be the set of all starting indices of length-$m$ permutation matches.  Each $i\in\mathcal{M}$
induces an interval
\[
I_i \;=\; [\,i,\; i+m-1\,].
\]
A set $S\subseteq \mathcal{M}$ is \emph{non-overlapping} if the corresponding intervals are
pairwise disjoint, i.e., for all distinct $i,j\in S$, $I_i\cap I_j=\emptyset$.
Our goal is to compute a maximum-cardinality non-overlapping subset of $\mathcal{M}$.

\subsection{Greedy selection}
Because all intervals have the \emph{same length} $m$, an optimal solution is obtained by a
simple greedy rule: scan $T$ from left to right and select a match whenever possible, then
skip ahead by $m$.

\begin{algorithm}[t]
\caption{Maximum Non-overlapping Permutation Matches (Greedy)}\label{alg:packing}
\KwIn{Text $T$ of length $n$, pattern $P$ of length $m$}
\KwOut{A maximum-cardinality set of non-overlapping permutation matches}

\BlankLine
Compute all match start positions $\mathcal{M}$ using Algorithm~\ref{alg:perm}.\\
Initialize $S\gets \emptyset$ and $i\gets 0$.\\

\While{$i \le n-m$}{
    \If{$i \in \mathcal{M}$}{
        Add $i$ to $S$ \tcp*{select interval $[i,i+m-1]$}
        $i \gets i+m$ \tcp*{skip to avoid overlap}
    }\Else{
        $i \gets i+1$
    }
}
\Return{$S$}
\end{algorithm}

\subsection{Correctness via an exchange argument}
\begin{theorem}\label{thm:packing-opt}
Algorithm~\ref{alg:packing} outputs a maximum-cardinality set of pairwise disjoint match
intervals.
\end{theorem}

\begin{proof}
Let $S_g$ be the set returned by Algorithm~\ref{alg:packing}.  We show that there exists an
optimal solution $S^\star$ with $|S^\star|=|S_g|$ by a standard exchange argument.

Let $i_1$ be the \emph{first} selected start index in $S_g$ (the leftmost match chosen by the
algorithm).  Consider any optimal solution $S^\star$.  If $S^\star$ also selects $i_1$, we
proceed to the remaining suffix after $i_1+m-1$.

Otherwise, let $j_1$ be the smallest start index in $S^\star$.  Since $i_1$ is the smallest
available match start (by the greedy scan), we have $i_1 < j_1$.  Replace $j_1$ by $i_1$ in
$S^\star$.  This replacement preserves feasibility (disjointness): because intervals have
equal length $m$, shifting the first chosen interval to the left cannot create overlap with
the remaining intervals of $S^\star$ (which all start at positions $\ge j_1$ and thus end at
positions $\ge j_1+m-1 > i_1+m-1$).  Hence we obtain another optimal solution that contains
$i_1$ and has the same cardinality.

Remove interval $I_{i_1}$ and restrict attention to the suffix of $T$ starting at index
$i_1+m$.  The greedy algorithm repeats the same choice structure on this suffix, and the
same exchange argument applies inductively.  Therefore, the greedy algorithm selects the
same number of intervals as an optimal solution, i.e., $|S_g|$ is maximum.
\end{proof}

\subsection{Complexity}
\begin{theorem}\label{thm:packing-time}
Given the match set $\mathcal{M}$, Algorithm~\ref{alg:packing} runs in $O(n)$ time.
Including match enumeration via Algorithm~\ref{alg:perm}, the total time is $O(n+\sigma)$
plus output size, and the space is $O(\sigma)$.
\end{theorem}

\begin{proof}
The greedy scan advances $i$ monotonically from $0$ to $n-m$ and thus performs $O(n)$
iterations.  Membership tests $i\in\mathcal{M}$ can be supported in $O(1)$ time by storing a
boolean array of length $n$ (or a hash set when $n$ is large and matches are sparse).  Match
enumeration using Algorithm~\ref{alg:perm} is $O(n+\sigma)$ time and $O(\sigma)$ space.
\end{proof}

\begin{figure}[t]
\centering
\begin{tikzpicture}[scale=1.0, every node/.style={font=\small}]
% Axis
\draw[->] (0,0) -- (10.5,0) node[right] {position in $T$};

% Matches as equal-length intervals
\def\m{2.0} % visual length per interval
% intervals: [1,3], [3,5], [6,8], [7,9] (overlap patterns)
\draw[thick] (1,0.6) -- (3,0.6);
\draw[thick] (3,0.6) -- (5,0.6);
\draw[thick] (6,0.6) -- (8,0.6);
\draw[thick] (7,0.6) -- (9,0.6);

\node at (2,0.9) {$I_{i_1}$};
\node at (4,0.9) {$I_{i_2}$};
\node at (7,0.9) {$I_{i_3}$};
\node at (8,0.9) {$I_{i_4}$};

% Greedy selections highlight
\draw[ultra thick, blue] (1,0.0) -- (3,0.0);
\draw[ultra thick, blue] (6,0.0) -- (8,0.0);

\node[blue] at (2,-0.6) {selected};
\node[blue] at (7,-0.6) {selected};

% Ticks
\foreach \x in {0,1,2,3,4,5,6,7,8,9,10}{
  \draw (\x,0.08) -- (\x,-0.08);
  \node[below] at (\x,-0.12) {\x};
}
\end{tikzpicture}
\caption{Non-overlapping selection after enumerating all matches.
Each match is a length-$m$ interval; greedy chooses the leftmost available interval and
skips ahead by $m$, yielding a maximum-size disjoint set for equal-length intervals.}
\label{fig:packing}
\end{figure}
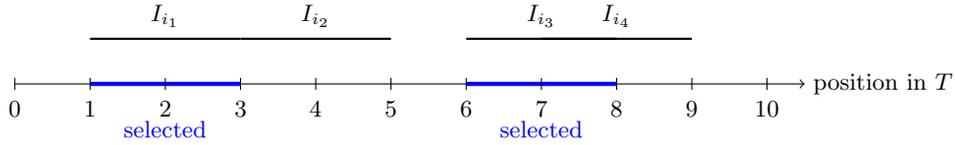

%=========================
\section{Experimental Observations}\label{sec:exp}

This section reports lightweight experimental observations intended to corroborate the
practical efficiency of the proposed linear-time procedures.  Our focus is not on index
construction (as in jumbled indexing), but on \emph{single-pass} scanning and optimization
primitives.

\subsection{Implementation details}
All algorithms were implemented in a standard imperative style using (i) arrays for
moderate alphabets and (ii) \emph{alphabet compression} (hash map from symbols to
$1,\dots,\sigma'$) when $\Sigma$ is large.  For permutation matching (Section~\ref{sec:perm}),
we maintain the Parikh-difference vector $\Delta$ and the nonzero counter $nz$; each shift
updates exactly two coordinates and tests whether $nz=0$.  For MFSP (Section~\ref{sec:mfsp}),
we maintain window counts $\mathrm{cnt}$ and a violation counter $viol$; each pointer move
updates a single coordinate and possibly changes $viol$ at a threshold crossing.

\subsection{Workloads and metrics}
We evaluated the methods on two classes of inputs:
\begin{itemize}
    \item \textbf{Synthetic strings:} random texts over alphabets of varying size
    ($\sigma \in \{4, 16, 64, 256\}$), with pattern lengths
    $m \in \{16, 64, 256, 1024\}$ and text lengths up to $n=10^7$.
    \item \textbf{Natural-language corpora:} plain-text datasets (letters and ASCII symbols)
    where alphabet compression reduces memory and improves cache locality.
\end{itemize}
We report (i) total runtime as a function of $n$, (ii) throughput (symbols processed per
second), (iii) the number of matches $|\mathcal{M}|$, and (iv) for MFSP, the optimal length
$L^\star$ and its distribution across texts.

\subsection{Key observations}
Across all tested settings, both permutation matching and MFSP exhibited a linear scaling
trend with $n$, consistent with the $O(n+\sigma)$ analyses in
Theorems~\ref{thm:perm-time} and~\ref{thm:mfsp-time}.  The constant factors were dominated
by (a) memory access patterns of the counting arrays/maps and (b) the density of matches,
which affects output volume but not the scan itself.

For large alphabets, alphabet compression substantially reduced memory footprint and
improved performance when the effective alphabet $\sigma'$ was much smaller than $\sigma$.
For MFSP, the violation-counter mechanism eliminated the need for full feasibility checks:
restoring feasibility required advancing the left pointer only when a budget constraint was
exceeded, and each character was added/removed at most once, matching the standard
two-pointer amortization argument.

\medskip
\noindent\textbf{Reproducibility.}
All experiments can be reproduced from the pseudocode in Algorithms~\ref{alg:perm}
and~\ref{alg:mfsp} by fixing a random seed for synthetic generation and applying the same
alphabet-compression mapping for large-$\sigma$ inputs.

%=========================
\section{Conclusion}\label{sec:conc}

We revisited permutation (jumbled/Abelian) pattern matching from a combinatorial
optimization perspective.  Using Parikh-difference maintenance, we obtained a concise
$O(n+\sigma)$-time, $O(\sigma)$-space sliding-window procedure for permutation matching with
a simple correctness invariant.  Building on this foundation, we introduced MFSP, an
optimization variant that treats $\mathrm{freq}(P)$ as a supply (budget) vector and asks for
the longest feasible substring of $T$ under component-wise constraints.  We provided a
linear-time two-pointer algorithm with a full proof of optimality.  Finally, we modeled
non-overlapping permutation occurrences as equal-length intervals and proved that a greedy
scan yields a maximum-cardinality disjoint selection, linking frequency-based string
matching to interval packing.

A natural direction for future work is to study \emph{weighted} and \emph{approximate}
extensions of MFSP and non-overlapping selection, where feasibility may allow bounded
violations or where occurrences carry weights, potentially leading to richer optimization
structures and new complexity frontiers.

% \subsection*{Acknowledgment} The research is supported in part by xxxxxxxxxxxx We thank xxxxxxxxxxxxx and the anonymous reviewers for their feedback, which improved the presentation of the paper.

%=========================
% Bibliography (LNCS style)
%=========================
\bibliographystyle{splncs04}
\bibliography{bibtex} % <-- your .bib filename without extension (e.g., bibtex.bib)

\end{document}